\documentclass[11pt]{article}
\usepackage{color}
\usepackage{graphicx}
\usepackage{amsmath}
\usepackage{amssymb}
\usepackage{mathrsfs}
\usepackage[numbers,sort&compress]{natbib}
\usepackage{amsthm}
\numberwithin{equation}{section}
\usepackage{pgf}
\usepackage{graphicx}
\usepackage{tikz}
\usepackage{stmaryrd}
\usepackage[latin1]{inputenc}
\usepackage[T1]{fontenc}
\usepackage[english]{babel}
\usepackage{listings}
\usepackage{xcolor,mathrsfs,url}
\usepackage{amssymb}
\usepackage{amsmath}
\usepackage{ifthen}
\usepackage {caption}
\usepackage{xcolor}
\usepackage{cases}

\usepackage{indentfirst}

\usepackage{amsmath}
\usepackage{amsthm}
\usepackage{todonotes}
\usepackage{float}
\usepackage{subfigure}
\usepackage{hyperref}
\hypersetup{
	colorlinks=true, %set true if you want colored links
	linktoc=all,     %set to all if you want both sections and subsections linked
	linkcolor=blue,  %choose some color if you want links to stand out
	citecolor=blue%choose some color if you want links to stand out
}

\def\R{\mathbb R}
\def\C{\mathbb C}
\def\d{\rm d}
\def\Wr{\rm Wr}
\def\Im{\rm Im}
\def\exp{\rm exp}
\DeclareMathOperator*{\Res}{Res}

\newtheorem{lemma}{Lemma}[section]

\begin{document}
	
	\title{Focusing and defocusing mKdV equation with fully asymmetric nonzero boundary conditions: the inverse scattering transform  }
	\author{Yi ZHAO$^1$\thanks{\ Corresponding author and email address: zhaoy23@bnu.edu.cn },  Dinghao ZHU$^1$ }
	\footnotetext[1]{ \  School of Mathematical Sciences, Laboratory of Mathematics and Complex Systems, MOE, Beijing Normal University, 100875 Beijing, China }

	\date{ }
	
	\maketitle
	\begin{abstract}
		\baselineskip=15pt
		In this paper, we investigate the inverse scattering transform(IST) for the focusing and defocusing mKdV equation with fully asymmetric nonzero boundary conditions. Our analysis focuses on the properties of the Jost function, allowing us to establish the associated Riemann-Hilbert (RH) problem for both the focusing and defocusing cases. Furthermore, we recover the potential from the RH problem, which plays a crucial role in calculating the long-time asymptotic behavior of the solution.  
			In the symmetric case, the eigenvalues are given by $\lambda^2 = k^2 \pm q_0^2$. However, in the asymmetric case, the eigenvalues are represented as $\lambda_\pm^2 = k^2 \pm q_\pm^2$. As a result, we directly handle the branch cut instead of utilizing the four-sheeted Riemann surface. Naturally, this approach leads to the discontinuity of the functions across their respective branch cuts, which significantly impacts the entire development of the IST. \\[6pt]
		\\[6pt]

	\end{abstract}

	\baselineskip=16pt
	
	\newpage

	\tableofcontents

	\section{ Introduction}
	\indent\setlength{\parindent}{1em}
	In this paper, we study the Cauchy problem to the modified Korteweg-de Viries(mKdV) equation
	\begin{equation}
		\begin{split}
			&q_t-6\sigma q^2q_x+q_{xxx}=0,\\
			&q\rightarrow q_{\pm},\quad x\rightarrow \pm\infty,
			\label{equ}
		\end{split}
	\end{equation}
	where $q(x,t)\in\R$, $\sigma=\pm 1$ and $q_\pm$ are real numbers. Without loss of generality, in the subsequent content, we assume that 
	$q_->q_+>0$. 
	
	The Korteweg-de Vries(KdV) equation 
	\begin{equation}
		u_t+6uu_x+u_{xxx}=0
		\label{kdv}	
	\end{equation}
	was proposed by Korteweg and de Vries to describe wave propagation on the surface of shallow water\cite{as1981siam}. 
	
	It is well-known that 	the mKdV equation  (\ref{equ})  admits the Lax pair\cite{lax}
	\begin{align}
		&\phi_x=X\phi,\label{spec}\\
		&\phi_t=T\phi,\label{time}
	\end{align}
	where 
	\begin{equation}
		\begin{split}
			&X=ik\sigma_3+Q,\\
			&T=4k^2X-2ik\sigma_3(Q_x-Q^2)+2Q^3-Q_{xx},
		\end{split}
		\nonumber
	\end{equation}
	with $Q=\left(\begin{array}{cc}
		0 & q(x,t)\\
		\sigma q(x,t) & 0
	\end{array}\right).$
	
	In 1967, Miura presented an explicit nonlinear transformation(the well-known Miura transformation) \cite{miura1968} which related  solutions of the KdV equation (\ref{kdv}) and the mKdV equation (\ref{equ}).
	
	The mKdV equation (\ref{equ}) arises in various physical fields, such as acoustic waves, phonons in a anharmonic lattices\cite{3,4}, Alfvén wave in cold collision-free plasma\cite{5,6}, dynamics of traffic flow \cite{9,10} and meandering ocean currents\cite{8}. 
	
    There has been significant research on the mKdV equation (\ref{equ}) in the field of integrable systems\cite{xin1,xin2,xin3,xin4,xin5}. The inverse scattering transform(IST) has been widely used to obtain the N-soliton solution for the mKdV equation (\ref{equ}). Under  zero boundary conditions(ZBCs), the N-soliton solution has been explicitly derived\cite{25,wadati}. Under  nonzero boundary conditions(NZBCs),  Yeung et al obtained N-soliton solution for mKdV equation (\ref{equ}) with nonzero vaccum parameter\cite{Yeung1,Yeung2}, while Alejo obtained the explicit breather solution for the focusing mKdV equatio\cite{Alejo}.   For both focusing and defocusing mKdV equation with symmetric NZBCs, Zhang and Yan  developed the direct and inverse scattering problem using a suitable uniformization variable \cite{zhangyan}. He et al reexamined  the Hamiltonian formalism with NZBC\cite{He}. 
    Additionally, several studies have explored the asymptotic analysis of the mKdV equation (\ref{equ}). The long-time asymptotic behavior with ZBCs was derived by the steepest descent method\cite{longz}.  Baldwin discovered the  long-time-asymptotic dispersive shock wave solution for the mKdV equation in 2013\cite{baldwin}. In 2016, Germain proved a full asymptotic stability result for solitary wave solutions of the mKdV equation\cite{Germain}.
	
	To the best of our knowledge, no results are available for the mKdV equation with fully asymmetric boundary conditions.
	
%	The paper is organized as follows. In section \ref{sec:section2}, for the focusing mKdV equation, without making use of the Riemann surface, insteaded by dealing the branch cut directly, we establish the direct and inverse scattering transformation and recover the potential from the RH problem.
%	In section \ref{sec:section3}, for the defocusing mKdV equation, we present the different branch cut which leads to a main difference from the focusing case that is the analyticity of the eigenfunction. Also, due to the self-adjoint operator of the defocusing case, there exhibits a different discrete spectrum. Finally, we also obtain the expression for the potential. 
	
	The paper is organized as follows.
	In section \ref{sec:section2}, we focus on the focusing mKdV equation. Instead of using the Riemann surface approach, we directly handle the branch cut to establish the direct and inverse scattering transformation. We solve the Riemann-Hilbert (RH) problem and recover the potential from it.
	In section \ref{sec:section3}, we consider the defocusing mKdV equation. We present a different branch cut configuration, which leads to a key difference compared to the focusing case: the analyticity of the eigenfunction. Additionally, the self-adjoint operator property of the defocusing case results in a distinct discrete spectrum. We obtain the expression for the potential in this case as well.

	\section{The focusing mKdV equation}\label{sec:section2}
		In this  section, we consider the focusing case, i.e. $\sigma=-1$
	First, let us establish and define some notations that will be used throughout this section.
	\begin{equation}\begin{split}
			&\Sigma_{\pm}=[-i q_\pm, iq_\pm],\\
			& \Sigma_0=[-iq_-,-iq_+]\cup [iq_+,iq_-],      \\
			&\mathring{\Sigma}_{\pm}=(-i q_\pm, iq_\pm),        \\
			&\mathring{\Sigma}_{0}= (-iq_-, -iq_+)\cup (iq_+, iq_-).     
			\nonumber
	\end{split}\end{equation}
		And three Pauli matrices are defined as
	$$\sigma_1=\left(
	\begin{array}{cc}
		0 & 1\\
		1 & 0
	\end{array}
	\right),\quad \sigma_2=\left(
	\begin{array}{cc}
		0 & -i\\
		i & 0
	\end{array}
	\right),\quad \sigma_3=\left(
	\begin{array}{cc}
		1 & 0\\
		0 & -1
	\end{array}
	\right).
	$$
	
		\subsection{Direct scattering problem}
	Next, we  will develop the spectral analysis and introduce the scattering matrix to prepare for the set-up of the RH problem.
	\subsubsection{Spectral analysis}

	To obtain the Jost functions of the spectral problem (\ref{spec}), we need to diagonalize the asymptotic matrix $X_{\pm}=ik\sigma_3+Q_{\pm}$. By calculation, the eigenvalues of the  matrix $X_{\pm}$ are $\pm i \lambda_{\pm}$, where
	
	$$\lambda_{\pm}^2=k^2+q_{\pm}^2.$$
	
	To avoid the multivalues  of the function $\lambda_{\pm}$, we take the branch cut for $\lambda_{\pm}$ on $\Sigma_{\pm}$. For all $k\in\C\backslash\Sigma_{\pm}$, we define $\lambda_{\pm}$ as the single-valued, analytic and continuous as $k$ approaches $\Sigma_\pm$ from right. For $k\in\Sigma_\pm$,  we take the positive value of the real square root for $\lambda_\pm$. 
	\begin{center}
		\begin{tikzpicture}
			%\draw [fill=pink,ultra thick,pink] (-3.8,0) rectangle (-0.8,1);
			%\draw [fill=cyan,ultra thick,cyan] (-3.8,0) rectangle (-0.8,-1);
			\draw [->](-6.2,0)--(1.8,0);

			\draw [->](-2.2,-3.5)--(-2.2,3.7);

			\draw [-](-2.2,2)--(-2.1,2);
			\node [thick] [left]  at (-2.2,2){\footnotesize $iq_-$};
			
			\draw [-](-2.2,1)--(-2.1,1);
			\node [thick] [left]  at (-2.2,1){\footnotesize $iq_+$};
			
			\draw [-](-2.2,-1)--(-2.1,-1);
			\node [thick] [left]  at (-2.2,-1){\footnotesize $-iq_+$};
			
			\draw [-](-2.2,-2)--(-2.1,-2);
			\node [thick] [left]  at (-2.2,-2){\footnotesize $-iq_-$};
			
			\draw[dashed] [blue] [ultra thick]  (-2.1,-1)--(-2.1,1);
			\node[blue] [ultra thick] [right]  at (-2.0,1){\footnotesize $\Sigma_+$};
			
			\draw[dashed] [red] [ultra thick] (-2.3,-2)--(-2.3,2);
			\node[red] [ultra thick] [left]  at (-2.5,1.7){\footnotesize $\Sigma_-$};
			%		
			%		\draw[dashed] [red]  (-6.2,-0.1)--(-3.2,-0.1);
			%		\draw[dashed] [red]  (-1.2,-0.1)--(1.8,-0.1);
			%		\node[red] [thick] [above]  at (-1.2,-1.1){\footnotesize $\Sigma_+$};
			%		
			%		\node [thick] [above]  at (-1.5,-1.5){\footnotesize $k_n$};
			%		\node [thick] [above]  at (-1.4,-1.1){\footnotesize $\centerdot$};
			%		
			%		\node [thick] [above]  at (-2.9,0.8){\footnotesize $-k_n$};
			%		\node [thick] [above]  at (-2.9,0.7){\footnotesize $\centerdot$};
			%		
			%		\node [thick] [above]  at (-2.9,-1.4){\footnotesize $-k_n^*$};
			%		\node [thick] [above]  at (-2.9,-1.1){\footnotesize $\centerdot$};
			%\draw [ultra thick,red](-2.2,-0.6)--(-2.2,0.6);
			%\node at (-2.6, -0.6 )  {$-iq_0 $};
			%\node at (-2.5, 0.6 )  {$iq_0 $};
			%\node at (-3.4, 0.7 )  {$  S_1$};
			%\node at (-1.5, 0.7 )  {$  +$};
			%\node at (-1.5, -0.7 )  {$  -$};
			%\node at (-2.9, 0.7 )  {$  -$};
			%\node at (-2.9, -0.7 )  {$  +$};
			%\node [thick] [above]  at (-1.5, 0.3){\footnotesize ${\rm Im}k>0$};
			%\node [thick] [above]  at (-1.5, -0.8){\footnotesize ${\rm Im}k<0$};
			\node [thick] [above]  at (-2.2,3.7){\footnotesize ${\rm Im}k$};
			\node [thick] [above]  at (2.1,-0.2){\footnotesize ${\rm Re}k$};
		\end{tikzpicture}
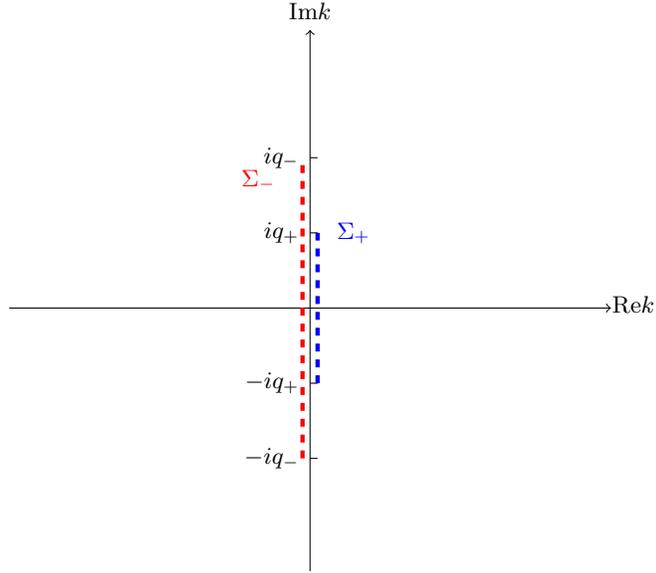
		\center{\footnotesize {\bf Fig.~1.} The branch cuts $\Sigma_-$ and $\Sigma_+$ to the focusing case}
	\end{center}

	The  eigenvector matrix is 
	$$E_{\pm}=I+\frac{i}{k-\lambda_{\pm}}\sigma_3 Q_{\pm}$$ corresponds to the eigenvalues $\mp i\lambda_{\pm}$, 
	and the determinant is $\det E_\pm=\frac{2\lambda_\pm}{\lambda_\pm-k}\triangleq \gamma_\pm(k)$.
	
	After diagonalizition,  we have 
	\begin{equation}
		\phi_{\pm}\sim E_{\pm}e^{-i\lambda_{\pm}x\sigma_3},\quad x\rightarrow \pm\infty.
	\end{equation}
	Define the normalized Jost functions $\mu_{\pm}$ as
	\begin{equation} 
		\mu_{\pm}=\phi_{\pm}e^{i\lambda_{\pm}x\sigma_3},
		\label{2}
	\end{equation}
	then we have $$\mu_{\pm}\rightarrow E_{\pm},\quad x\rightarrow \pm \infty.$$ 
	And they satisfy the Volterra integral equation
	\begin{equation}
		\mu_{\pm}=E_{\pm}(k)+\int_{\pm\infty}^x E_{\pm}(k)e^{-i\lambda_{\pm}(x-y)\hat\sigma_3} \left[
		E_{\pm}^{-1}(k)\Delta Q_{\pm}(y) \mu_{\pm}(y,k) 
		\right] \d y
		\label{integral}
	\end{equation}
	with $\Delta Q_{\pm}=Q-Q_{\pm}$.
	
	\begin{lemma}\label{2.1}  If $(1+|x|)(q-q_{\pm})\in L^1(\R_{\pm})$, then
		\begin{itemize}
			
			\item Analyticity: the function $\mu_{-,1}(x,t,k)$ and $\mu_{+,2}(x,t,k)$ are analytic in $\C^+\backslash [0,iq_-]$, whereas the function $\mu_{+,1}(x,t,k)$ and $\mu_{-,2}(x,t,k)$ are analytic in $\C^-\backslash [-iq_-,0]$. 
			%	And the eigenfunctions are well-defined on the corresponding branch cuts.
			%	
			%	\item The first symmetry:
			%	\begin{equation}
				%		\phi_\pm(x,t,k)=\sigma_2 \overline{\phi_\pm(x,t,\bar k)}\sigma_2,
				%	\end{equation}
			\item Symmetry: Let  $\widetilde \phi_{\pm}=\phi_{\pm}(x,t,-\lambda_{\pm}(k))$, then we have
			\begin{equation}
				\phi_{\pm}=\widetilde \phi_{\pm}\frac{i}{k-\lambda_{\pm}}\sigma_3 Q_{\pm},\quad k\in\mathring{\Sigma}_{\pm},
			\end{equation}
			
			Column-wise reads
			\begin{align}
				&\phi_{\pm,1}=\frac{iq_{\pm}}{k-\lambda_\pm}\widetilde \phi_{\pm,2}, \quad k\in\mathring{\Sigma}_{\pm},\label{22.6}\\
				&\phi_{\pm,2}=\frac{iq_{\pm}}{k-\lambda_\pm}\widetilde \phi_{\pm,1}, \quad k\in\mathring{\Sigma}_{\pm}.\label{22.7}
			\end{align}
			
			%	\item Asymptotics:
			%	\begin{align}
				%		&\mu_\pm(x,t,k)=I+O(\frac 1k), & k\rightarrow\infty, k\in \C^+,\\
				%		&\mu_\pm(x,t,k)=\frac{i}{k-\lambda_\pm}\sigma_3Q_\pm+O(1),  \hspace{0.2cm} &k\rightarrow\infty, k\in \C^-.
				%		\end{align}
		\end{itemize}
	\end{lemma}
	\begin{proof}
		%	It suffices to prove the statement for one Jost function, for example, for $\mu_-$. Note that for all $k\in\C\backslash\Sigma_\pm$, we have $\Im \lambda_\pm>0$. Thus by standard Neumann series method, the analyticity follows.  And the function $\mu_\pm(x,t,k)$ is well-defined in $\Sigma_\pm$. In fact, for $k\neq \pm iA_\pm$, we have 
		%	\begin{equation}
			%		E_-e^{-i\lambda_-(x-y)\sigma_3}E_-^{-1}=\frac{1}{\lambda_-}\sin(\lambda_-(x-y))X_-+\cos(\lambda_-(x-y))I,
			%	\end{equation}
		%	Taking the limit $k\rightarrow  \pm iA_\pm$ in the integral equation (\ref{integral}) leads to 
		%	\begin{equation}\begin{split}
				%		\mu_-(x, t, \pm iA_- )=&I\mp \frac{i}{A_-}\sigma_3Q_-+\int_{-\infty}^x [(x-y)X_-(\pm iA_-, t)+I]\\ &\cdot\Delta Q_-(y,t)\mu_-(y, t, \pm iA_-) \d y.
				%	\end{split}
			%	\end{equation}
		%The assumption $(1+|x|)(q-q_{\pm})\in L^1(\R_{\pm})$  ensures that  the function $\mu_-(x, t, k)$ is well-defined in $\Sigma_-$.
		
		We only prove the  symmetry.  Note that in defining $\lambda_\pm$, we could have taken the opposite sign of  the comlex square roots.  With some abuse of notation, we temporaily express the dependence of the Jost functions on the choice of sign explicitly. 
		
		From the above discussion, if $\phi_\pm$ solves the spectral problem (\ref{spec}), so does $\widetilde \phi_{\pm}$. Recall the asymptotic $\phi_\pm\sim \left( I+\frac{i}{k-\lambda_\pm}\sigma_3Q_\pm  \right) e^{-i\lambda_\pm x\sigma_3}$ as $x\rightarrow \pm\infty$, thus we have $\widetilde \phi_\pm \sim  \left( I+\frac{i}{k+\lambda_\pm}\sigma_3Q_\pm  \right) e^{i\lambda_\pm x\sigma_3}$. Direct check leads to 
		\begin{equation}
			\phi_\pm=\widetilde \phi_\pm \frac{i}{k-\lambda_\pm}\sigma_3 Q_\pm, \quad k\in \mathring \Sigma_\pm,
		\end{equation}
		which can be rewittten column-wise as
		\begin{align}
			& \phi_{\pm, 1}=\frac{iq_\pm}{k-\lambda_\pm}\widetilde \phi_{\pm, 2},\quad k\in\mathring \Sigma_\pm,  \\
			&\phi_{\pm, 2}=\frac{iq_\pm}{k-\lambda_\pm}\widetilde \phi_{\pm, 1},\quad k\in\mathring \Sigma_\pm.
		\end{align}
	\end{proof}
	
	\subsubsection{Scattering matrix}
	There exists a scattering matrix $S(k)$ such that 
	\begin{equation}
		\phi_-(x,t,k)=\phi_+(x,t,k)S(k),\quad k\in\Sigma_+.
		\label{scattering}
	\end{equation}
	Thus, we have $\det S=\frac{\gamma_-}{\gamma_+}$. And the  scattering coefficients can be written as the Wronskians expressions
	\begin{equation}
		\begin{split}
			&s_{11}=\frac{\Wr(\phi_{-,1},\phi_{+, 2})}{\gamma_+},\quad  s_{12}=\frac{\Wr(\phi_{-,2},\phi_{+, 2})}{\gamma_+},\\
			&s_{21}=\frac{\Wr(\phi_{+,1},\phi_{-, 1})}{\gamma_+},\quad  s_{22}=\frac{\Wr(\phi_{+,1},\phi_{-, 2})}{\gamma_+}.
		\end{split}\label{wron}
	\end{equation}
	Define
	\begin{equation}
		\rho=\frac{s_{21}}{s_{11}},\quad \tilde \rho=\frac {s_{12}} {s_{22}}.
	\end{equation}
	\begin{lemma}\label{2.2}
		%		The scattering coefficients satisfy the following properties:
		%		\begin{itemize}
			%			\item $s_{11}(k)$ is analytic in $\C\backslash\Sigma_{-}$, and 
			%			\item  
			Let $s_{11}^-(k)=\lim_{\epsilon\uparrow 0} s_{11}(\epsilon+k)$ for $k\in\Sigma_-$, then we have 
			\begin{equation}
				\begin{split}
					&s_{11}^-(k)=s_{22}(k)\frac{( k+\lambda_+ )(k-\lambda_-)}{  -q_+q_- },\quad k\in\Sigma_+,\\
					&s_{11}^-(k)=s_{12}(k)\frac{k-\lambda_-}{iq_-},\quad k\in\Sigma_0.
				\end{split}\label{2.17}
			\end{equation}
			%			\item $$S(k)=I+O(\frac 1k),\quad k\rightarrow\infty, k\in\C^+,$$
			%			$$S(k)=\frac{q_+}{q_-}I+O(k),\quad k\rightarrow\infty, k\in\C^-.$$
			%\end{itemize}
			Let $s_{22}^-(k)=\lim_{\epsilon\uparrow 0} s_{22}(\epsilon+k)$ for $k\in\Sigma_-$, then we have 
			\begin{equation}
				\begin{split}
					&s_{22}^-(k)=s_{11}(k)\frac{( k+\lambda_+ )(k-\lambda_-)}{  -q_+q_- },\quad k\in\Sigma_+,\\
					&s_{22}^-(k)=s_{21}(k)\frac{k-\lambda_-}{iq_-},\quad k\in\Sigma_0.
				\end{split}\label{2.17}
			\end{equation}
		\end{lemma}
		\begin{proof}
			For $k\in \Sigma_+$, we know  $\gamma_+^-\triangleq \lim_{\epsilon\uparrow 0}\gamma_+(k+\epsilon)=\frac{2\lambda_+}{k+\lambda_+}$. It follows from the Wronskians expression of $s_{11}$ (\ref{wron}), (\ref{22.6}) and (\ref{22.7}) that 
			\begin{equation}
				\begin{split}
					s_{11}^-=&\frac{\Wr (\phi_{-,1}^-, \phi_{+,2}^-)}    { \gamma_+^- }\\
					=&\frac{k+\lambda_+ }  {2\lambda_+} \Wr \left(   \frac{k-\lambda_-}{iq_-} \phi_{-,2},  \frac{k-\lambda_+}   {iq_+}\phi_{+,1}       \right)\\
					=&s_{22}  \frac{( k+\lambda_+ )(k-\lambda_-)}{  -q_+q_- }.
				\end{split}
			\end{equation}
			For $k\in\Sigma_0$,  the function $\phi_{+,2}$ has no jump, i.e, $\phi_{+,2}^-=\phi_{+,2}$, thus we have 
			\begin{equation}
				\begin{split}
					s_{11}^-&=\frac{\lambda_+-k}  {2\lambda_+}\Wr (\phi_{-,1}^-, \phi_{+,2})\\
					&=\frac{\lambda_+-k}  {2\lambda_+}\Wr \left( \frac{k-\lambda_-}{iq_-} \phi_{-,2},    \phi_{+,2} \right)\\
					&=s_{12}\frac{k-\lambda_-}  {iq_-}.
				\end{split}
			\end{equation}
			Performing a similar procedure leads to 
			\begin{equation}
				\begin{split}
					s_{22}^-=&\frac{\Wr (\phi_{+,1}^-, \phi_{-,2}^-)}    { \gamma_+^- }\\
					=&\frac{k+\lambda_+ }  {2\lambda_+} \Wr \left(   \frac{k-\lambda_+}{iq_+} \phi_{+,2},  \frac{k-\lambda_-}   {iq_-}\phi_{-,1}       \right)\\
					=&s_{11}  \frac{( k+\lambda_+ )(k-\lambda_-)}{  -q_+q_- },\quad k\in\Sigma_+
				\end{split}
			\end{equation}
			and 
			\begin{equation}
				\begin{split}
					s_{22}^-&=\frac{\lambda_+-k}  {2\lambda_+}\Wr \left(\phi_{+,1}, \frac{k-\lambda_-}{iq_-}\phi_{-,1}\right)\\
					&=\frac{\lambda_+-k}  {2\lambda_+}\Wr \left( (\phi_{+,1},  \frac{k-\lambda_-}{iq_-} \phi_{-,1} \right)\\
					&=s_{21}\frac{k-\lambda_-}  {iq_-},\quad k\in\Sigma_0.
				\end{split}
			\end{equation}
			
		\end{proof}
		\subsubsection{Discrete spectrum and residue condition}
		
		Suppose $s_{11}(k,t)$ has $J$ simple zeros in $\C^+\backslash[0,iq_-]$ denoted by $k_j,\hspace{0.1cm} j=1,2,\ldots, J$. From the Wronskians expression of $s_{11}(k,t)$, there exists a constant only depending on $t$ such that
		\begin{equation}
			\phi_{-,1}(x,t,k_j)=b_j(t)\phi_{+,2}(x,t,k_j).
		\label{norming}	\end{equation}
		On account of (\ref{2}), we have
		$$\mu_{-,1}(x,t,k_j)=b_j(t)\mu_{+,2(x,t,k_j)}e^{i(\lambda_{+,j}+\lambda_{-,j})x},$$
		thus for $j=1,2,\ldots,J$, we get
		\begin{equation}
			\Res_{k=k_j}\bigg[   \frac{\mu_{-,1}(x,t,k)}{s_{11}(k,t)}    \bigg]=\frac{\mu_{-,1}(k_j)}{s_{11}^{'}(k_j)}=C_j(t)e^{i(\lambda_{+,j}+\lambda_{-,j})x}\mu_{+,2}(x,t,k_j),
		\end{equation}
		with $C_j=\frac{b_j}{s_{11}^{'}(k_j)}$ and $\lambda_{\pm,j}=\lambda_\pm(k_j)$.
		%	Therefore, we have
		%	\begin{equation}
			%		\Res_{k=k_j} M(x,t,k)= \bigg(   
			%		C_j e^{i(\lambda_{+,j}+\lambda_{-,j})x} M_2(k_j) \hspace{0.1cm}, \hspace{0.1cm} 0
			%		\bigg).
			%		\end{equation}
		
		From the symmetry of the scattering matrix(see (35) in \cite{zhangyan}), we know that $k_j^*$ are also zeros of $s_{22}(k)$. Similarly, we get
		\begin{equation}
			\Res_{k=k_j^*} \bigg[   \frac{\mu_{+,2}(x,t,k)}{s_{22}(k,t)}    \bigg]=\frac{\mu_{+,2}(k_j^*)}{s_{22}^{'}(k_j^*)}=-C_j^*e^{i(\lambda_{+,j}+\lambda_{-,j})x}\mu_{-,1}(x,t,k_j^*),
		\end{equation}
		
		\subsection{Time evolution} \label{timeevolution}
		Define the solution of the both parts of the Lax pair (\ref{spec})(\ref{time}) as
		\begin{equation}
			\psi(x,t,k)=\phi_\pm(x,t,k)C_\pm(k,t),
			\end{equation}
		where the function $C$ is independent of $x$.
		Taking the derivative with respect to $t$ on  both sides gives
		$$(C_{\pm})_t=R_\pm C_\pm$$
		where $R_\pm =\phi_\pm^{-1}(T\phi_\pm-(\phi_{\pm})_t)$. 
		Since the function $C$ is independent of $x$, the function $R_\pm$ is also independent of $x$. Therefore we have
		\begin{equation}
			R_\pm=\lim_{x\rightarrow\pm\infty}\phi_\pm^{-1}(T\phi_\pm-(\phi_{\pm})_t)=if_\pm(k)\sigma_3,
			\end{equation}
		where $$f_\pm(k)=2(q_\pm^2-2k^2)\lambda_\pm(k).$$
		Thus we can express $(\phi_{\pm})_t$ as
		\begin{equation}
			(\phi_{\pm})_t=T\phi_\pm-\phi_\pm R_\pm.
			\end{equation}
		Taking the derivative with respect to $t$ of Eq. (\ref{scattering}) gives 
		$$S_t=R_+S-SR_-,$$
		element-wise reads
		\begin{equation}
			\begin{split}
			&s_{11}(k,t)=s_{11}(k,0)\exp[it(f_+(k)-f_-(k))],\\
			&s_{21}(k,t)=s_{21}(k,0)\exp[-it(f_+(k)+f_-(k))],
			\end{split}
			\end{equation}
		which result in 
		$$\rho(k,t)=\rho(k,0)\exp(-2itf_+(k))$$
		and 
		$$s_{11}^{'}(k_j,t)=s_{11}^{'}(k_j,0)\exp[it(f_+(k_j)-f_-(k_j))].$$
		Taking the derivative with respect to $t$ of Eq. (\ref{norming}), we have
		$$b_j(t)=b_j(0)\exp[-it(f_+(k_j)+f_-(k_j))],$$
		therefore we obtain
		$$C_j(t)=C_j(0)\exp(-2if_+(k_j)t).$$
		\subsection{Inverse scattering problem}
		Next, we will establish the RH problem associated to the focusing modified KdV equation (\ref{equ}) and recover the potential.
		\subsubsection{Riemann-Hilbert problem}
		
		Define a meromorphic function as 
		
		\begin{equation}
			M(x;k)=\left\{\begin{split}
				&\left[   \frac{\mu_{-,1}}{s_{11}}\hspace{0.1cm},\hspace{0.1cm}\mu_{+,2}    \right],\quad&k\in\C^+\backslash [0,iq_-],\\
				&\left[  \mu_{+,1}  \hspace{0.1cm},\hspace{0.1cm} \frac{\mu_{-,2}}{s_{22}}     \right],\quad&k\in\C^-\backslash [-iq_-,0].
			\end{split}\right.\label{3.1}
		\end{equation}
		
		\noindent \textbf{Jump on  $\R$} 
		For $k\in\R$, we write the RH problem as  $M_+(x,k)=M_-(x,k)V_0(x,k)$, i.e.
		\begin{equation}
			\bigg(  \frac{\mu_{-,1}}{s_{11}}, \mu_{+,2}     \bigg)=\bigg(  \mu_{+,1}, \frac{\mu_{-,2}}{s_{22}}     \bigg)V_0,\quad k\in\R,
		\end{equation}
		where the subscripts $\pm$ denote limiting values from the upper/lower complex plane and  
		\begin{equation}
			V_0(x,k)=\left(\begin{array}{cc}
				(1-\rho\tilde\rho)e^{i(\lambda_--\lambda_+)x} & -\tilde\rho e^{-2i\lambda_+x}\\
				\rho e^{2i\lambda_-x} & e^{i(\lambda_--\lambda_+)x}
			\end{array}\right).
		\end{equation}
		
		\noindent \textbf{Jump on the $\Sigma_+^+$} 
		For $k\in\Sigma_+^+:=[0,iq_+]$, the RH problem can  be written as $M_+(x,k)=M_-(x,k)V_1(x,k)$, where the subscripts $\pm$ denote limiting values from the right/left complex plane, i.e. 
		\begin{equation}
			\bigg(  \frac{\mu_{-,1}}{s_{11}}, \mu_{+,2}     \bigg)=\bigg(  \frac{\mu_{-,1}^-}{s^-_{11}}, \mu_{+,2}^-     \bigg)V_1,\quad k\in\Sigma_+^+.
		\end{equation}
		Due to the scattering relation (\ref{scattering}) and the symmetris, we have
		\begin{equation}
			\begin{split}
				\frac{\phi_{-,1}  }{s_{11}}&= (1-\rho\tilde\rho)\phi_{+,1}+\rho \frac{\phi_{-,2}}{s_{22}}\\
				&=(1-\rho\tilde\rho) \frac{iq_+}  {k-\lambda_+ }\phi_{+,2}^- +\rho \frac{(k+\lambda_+) } {iq_+} \frac{\phi_{-,1}^-}{s_{11}^-}
			\end{split}
		\end{equation}
		and 
		\begin{equation}
			\begin{split}
				\phi_{+,2}&=\frac{\phi_{-,2}}{s_{22}}-\tilde \rho \phi_{+,1}\\
				&=\frac{ k+\lambda_+ }  {iq_+} \frac{\phi_{-,1}^-}{s_{11}^-} - \tilde \rho \frac{iq_+}{k-\lambda_+}\phi_{+,2}^-
			\end{split}
		\end{equation}
		which results in 
		\begin{equation}
			V_1=\left(\begin{array}{cc}
				\frac{k+\lambda_+}{iq_+}\rho e^{2i\lambda_-x}   &   \frac{ k+\lambda_+}  {iq_+}e^{i(\lambda_--\lambda_+)x} \\
				(1-\rho\tilde\rho) \frac{iq_+}  {k-\lambda_+ }e^{i(\lambda_--\lambda_+)x}   &  -\tilde \rho \frac{iq_+}{k-\lambda_+}e^{-2i\lambda_+x}
			\end{array}\right).
		\end{equation}
		\noindent \textbf{Jump on the $\Sigma_+^-$} 
		For $k\in \Sigma_+^-:= [-iq_+, 0]$, the jump can be determined by $M_+(x,k)=M_-(x,k)V_2(x,k)$, where the subscripts $\pm$ denote limiting values from the right/left complex plane,
		i.e.
		\begin{equation}
			\bigg(   \mu_{+,1} , \frac{\mu_{-,2}}{s_{22}}       \bigg)=\bigg(   \mu_{+,1}^- , \frac{\mu_{-,2}^-}{s^-_{22}}     \bigg)V_2,\quad k\in\Sigma_+^-.
		\end{equation}
		Due to (\ref{scattering}) and the symmetris, we have
		\begin{equation}
			\begin{split}
				\phi_{+,1}&=\frac{\phi_{-,1}}{s_{11}}-\rho\phi_{+,2}  \\
				&=\frac{k+\lambda_+}{iq_+} \frac{\phi_{-,2}^-}{s_{22}^-} - \rho \frac{iq_+}{k-\lambda_+}\phi_{+,1}^-
			\end{split}
		\end{equation}
		and 
		\begin{equation}
			\begin{split}
				\frac{\phi_{-,2}}{s_{22}}&=\tilde \rho \phi_{+,1} + \phi_{+,2}\\
				&=\tilde \rho \frac{\phi_{-,1}}{s_{11}} + (1-\rho\tilde \rho)\phi_{+,2}\\
				&=\tilde \rho \frac{k+\lambda_+}{iq_+}\frac{\phi_{-,2}^-}{s_{22}^-}+ (1-\rho\tilde \rho)\frac{iq_+}{k-\lambda_+}\phi_{+,1}^-,
			\end{split}
		\end{equation}
		then we have 
		\begin{equation}
			V_2=\left(\begin{array}{cc}
				- \rho \frac{iq_+}{k-\lambda_+}e^{2i\lambda_+x}  &  (1-\rho\tilde \rho)\frac{iq_+}{k-\lambda_+}  e^{i(\lambda_+-\lambda_-)x}  \\
				\frac{k+\lambda_+}{iq_+} e^{i(\lambda_+-\lambda_-)x}    & \tilde \rho \frac{k+\lambda_+}{iq_+}e^{-2i\lambda_-x}
			\end{array}\right).
		\end{equation}
		\noindent \textbf{Jump on the $\Sigma_0^+$} 
		For $k\in \Sigma_0^+:= [iq_+, iq_-]$, the jump can be determined by $M_+(x,k)=M_-(x,k)V_2(x,k)$, 	where the subscripts $\pm$ denote limiting values from the right/left complex plane,
		i.e.
		\begin{equation}
			\bigg(   \frac{\mu_{-,1}}{s_{11}}, \mu_{+,2}        \bigg)=\bigg(    \frac{\mu^-_{-,1}}{s^-_{11}}, \mu_{+,2}     \bigg)V_3,\quad k\in \Sigma_0^+,
		\end{equation}
		Due to (\ref{scattering}) and the symmetris, we have
		\begin{equation}
			\begin{split}
				\frac{\phi_{-,1}}{s_{11}}&= \phi_{+,1}+ \rho \phi_{+,2}\\
				&=\frac{\phi_{-,2}}{s_{12}}+ (\rho-\frac{1}{\tilde\rho})\phi_{+,2}\\
				&=\frac{\phi_{-,1}^-}{s_{11}^-}+(\rho-\frac 1 {\tilde\rho})\phi_{+,2},
			\end{split}
		\end{equation}
		that is
		\begin{equation}
			V_3=\left(\begin{array}{cc}
				e^{2i\lambda_-x} & 0 \\
				\big(\rho-\frac{1}{\tilde\rho}\big)e^{i(\lambda_--\lambda_+)x} & 1
			\end{array}\right)
		\end{equation}
		\noindent \textbf{Jump on the $\Sigma_0^-$} 
		For $k\in \Sigma_0^-:= [-iq_-, -iq_+]$, the jump can be determined by $M_+(x,k)=M_-(x,k)V_4(x,k)$, 	where the subscripts $\pm$ denote limiting values from the right/left complex plane, 
		i.e.
		\begin{equation}
			\bigg(   \mu_{+,1} , \frac{\mu_{-,2}}{s_{22}}       \bigg)=	\bigg(   \mu_{+,1} , \frac{\mu^-_{-,2}}{s^-_{22}}       \bigg)V_4,\quad k\in \Sigma_0^-.
		\end{equation}
		Due to (\ref{scattering}) and the symmetris, we have
		\begin{equation}
			\begin{split}
				\frac{\phi_{-,2}}{s_{22}}& =\tilde \rho\phi_{+,1}+ \phi_{+,2}\\
				&=(\tilde\rho-\frac{1}{\rho})\phi_{+,1}+\frac{\phi_{-,1}}{s_{21}} \\
				&=(\tilde\rho-\frac{1}{\rho})\phi_{+,1} - \frac{\phi_{-,2}^-}{s_{22}^-},
			\end{split}
		\end{equation}
		that is
		\begin{equation}
			V_4=\left(\begin{array}{cc}
				1 & (\tilde\rho-\frac{1}{\rho} )e^{-i(\lambda_-+\lambda_+)x} \\
				0   & -e^{-2i\lambda_-x}
			\end{array}\right).
		\end{equation}
		Summarize the results, we have
		
		\noindent\textbf{RH problem 1.} Find a $2\times 2$ matrix-valued function $ M(x;k)$ with the following properties:
		\begin{itemize}
			\item $M(x;\cdot):\ \C \backslash\Sigma\rightarrow \C ^{2\times 2}$ is analytic, where $\Sigma=\R\cup \Sigma_-$.
			\item The limits of $M(x;k)$ as k approaches $\Sigma$ from the positive and negative side  exist and are continuous on $\Sigma$, moreover, they satisfy
			\begin{equation}
				M_+(x;k)=M_-(x;k)V(x;k),\quad k\in\Sigma,
				\label{jump}\end{equation}
			where 
			$$ V=\left\{\begin{split}
				&V_0, \quad k\in\R,\\
				&V_1, \quad k\in\Sigma_+^+,\\
				&V_2, \quad k\in\Sigma_+^-,\\
				&V_3, \quad k\in \Sigma_0^+,\\
				&V_4, \quad k\in \Sigma_0^-,
			\end{split}
			\right.
			$$
			\item $M(x;k)$ satisfies the asymptotic(see (53) in \cite{zhangyan})
			$$M(x;k) \rightarrow I,\quad |k|\rightarrow\infty, k\in\C\backslash\Sigma.$$
			$$M(x;k) \rightarrow \frac i k \sigma_3Q_{-}+O(1), \quad k\rightarrow 0.$$
		\end{itemize}
		
		\subsubsection{Reconstruction formula}
		Using the Plemelj formula, we can express the solution $M$ as follows:
		\begin{equation}
			\begin{split}
				M(x,t,k)=&I+\frac i k \sigma_3Q_{-}+\sum_{j=1}^J \frac{1}{k-k_j}\Res_{k=k_j}M^++ \sum_{j=1}^J \frac{1}{k-k_j^*}\Res_{k=k_j^*}M^-\\
				&+\frac{1}{2\pi i}\int_{\Sigma}\frac{[M^-(V-I)](x,t,\zeta)}{\zeta-k}\d \zeta,\quad k\in\C\backslash\Sigma,
			\end{split}\label{2.41}
		\end{equation}
		
		By a standard argument, we can recover the potential $q(x,t)$ 
		\begin{equation}
			q(x,t)=-i\lim_{\substack{
					k\rightarrow\infty}}kM_{12}(x,t,k).
		\end{equation} 
		Substituting (\ref{2.41}) into it, we obtain
		\begin{equation}
			q(x,t)=q_-(t)-i\sum_{j=1}^J C_j^*(t)e^{i  (\lambda_{+,j}^*+\lambda_{-,j}^*)  x   }M_{11}(x,t,k_j^*)		
			+\frac{1}{2\pi }\int_{\Sigma}\bigg[   M^-(V-I)  
			\bigg]_{12}{\d}   k
		\label{rec}\end{equation}
	\noindent\textbf{Remark 2.1}  When $q_-=q_+=q_0$, the reconstruction formula (\ref{rec}) matches Eq. (65) in \cite{zhangyan}. 
	
	%	\noindent\textbf{Remark 2.2} 
		\section{The defocusing mKdV equation}\label{sec:section3}
		
		In this  section, we consider the defocusing case, i.e. $\sigma=1$. 
		First, let us establish and define some notations that will be used throughout this section.
		\begin{equation}\begin{split}
				&\Sigma_{\pm}=(-\infty,-q_{\pm}]\cup [q_{\pm},\infty),\\
				&\Sigma_0=[-q_-, -q_+]\cup [q_+, q_-],\\
				&\mathring{\Sigma}_{\pm}= (-\infty,-q_{\pm})\cup (q_{\pm},\infty),        \\
				& \mathring{\Sigma}_{0}= (-q_-, -q_+)\cup (q_+, q_-).    
				\nonumber    \end{split}\end{equation}
		
		\begin{center}
			\begin{tikzpicture}
				%\draw [fill=pink,ultra thick,pink] (-3.8,0) rectangle (-0.8,1);
				%\draw [fill=cyan,ultra thick,cyan] (-3.8,0) rectangle (-0.8,-1);
				\draw [->](-6.2,0)--(1.8,0);

				\draw [->](-2.2,-2.5)--(-2.2,2.5);

				\draw [-](-4.2,0)--(-4.2,0.08);
				\node [thick] [below]  at (-4.2,0){\footnotesize $-q_-$};
				
				\draw [-](-3.2,0)--(-3.2,0.08);
				\node [thick] [below]  at (-3.2,0){\footnotesize $-q_+$};
				
				\draw [-](-1.2,0)--(-1.2,0.08);
				\node [thick] [below]  at (-1.2,0){\footnotesize $q_+$};
				
				\draw [-](-0.2,0)--(-0.2,0.08);
				\node [thick] [below]  at (-0.2,0){\footnotesize $q_-$};
				
				\draw[dashed] [blue] [ultra thick] (-6.2,0.1)--(-4.2,0.1);
				\draw[dashed] [blue] [ultra thick] (-0.2,0.1)--(1.8,0.1);
				\node[blue] [ultra thick] [above]  at (-5.2,0.2){\footnotesize $\Sigma_-$};
				
				\draw[dashed] [red] [ultra thick]  (-6.2,-0.1)--(-3.2,-0.1);
				\draw[dashed] [red] [ultra thick] (-1.2,-0.1)--(1.8,-0.1);
				\node[red] [ultra thick] [above]  at (-1.2,-1.1){\footnotesize $\Sigma_+$};
				%		
				%		\node [thick] [above]  at (-1.5,-1.5){\footnotesize $k_n$};
				%		\node [thick] [above]  at (-1.4,-1.1){\footnotesize $\centerdot$};
				%		
				%		\node [thick] [above]  at (-2.9,0.8){\footnotesize $-k_n$};
				%		\node [thick] [above]  at (-2.9,0.7){\footnotesize $\centerdot$};
				%		
				%		\node [thick] [above]  at (-2.9,-1.4){\footnotesize $-k_n^*$};
				%		\node [thick] [above]  at (-2.9,-1.1){\footnotesize $\centerdot$};
				%\draw [ultra thick,red](-2.2,-0.6)--(-2.2,0.6);
				%\node at (-2.6, -0.6 )  {$-iq_0 $};
				%\node at (-2.5, 0.6 )  {$iq_0 $};
				%\node at (-3.4, 0.7 )  {$  S_1$};
				%\node at (-1.5, 0.7 )  {$  +$};
				%\node at (-1.5, -0.7 )  {$  -$};
				%\node at (-2.9, 0.7 )  {$  -$};
				%\node at (-2.9, -0.7 )  {$  +$};
				%\node [thick] [above]  at (-1.5, 0.3){\footnotesize ${\rm Im}k>0$};
				%\node [thick] [above]  at (-1.5, -0.8){\footnotesize ${\rm Im}k<0$};
				\node [thick] [above]  at (-2.2,2.5){\footnotesize ${\rm Im}k$};
				\node [thick] [above]  at (2.1,-0.2){\footnotesize ${\rm Re}k$};
			\end{tikzpicture}
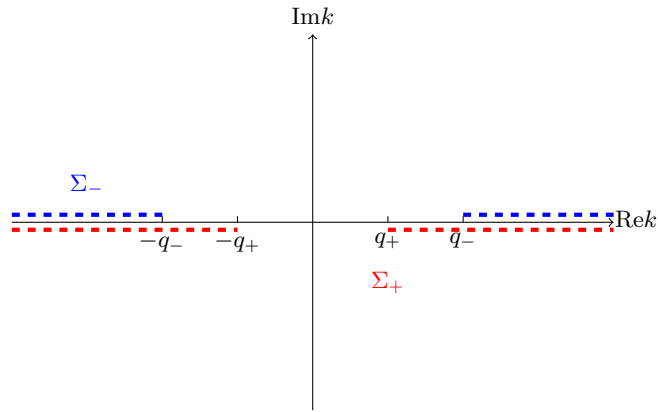
			\center{\footnotesize {\bf Fig.~2.} The branch cuts $\Sigma_-$ and $\Sigma_+$ to the defocusing case}
		\end{center}
		
		\subsection{Direct scattering problem}
			Next, we  will develop the spectral analysis and introduce the scattering matrix to prepare for the set-up of the RH problem.
		\subsubsection{Spectral analysis}

		To obtain the Jost functions of the spectral problem (\ref{spec}), we need to diagonalize the asymptotic matrix $X_{\pm}=ik\sigma_3+Q_{\pm}$. By calculation, the eigenvalues of the  matrix $X_{\pm}$ are $\pm i \lambda_{\pm}$, where
		$$\lambda_{\pm}^2=k^2-q_{\pm}^2.$$
		To avoid the multivalues  of the function $\lambda_{\pm}$, we take the branch cut for $\lambda_{\pm}$ on $\Sigma_{\pm}$. For all $k\in\C\backslash\Sigma_{\pm}$, we define $\lambda_{\pm}$ as the single-valued, analytic and continuous as $k$ approaches $\Sigma_\pm$ from above. For $k\in\Sigma_\pm$,  we take the positive value of the real square root for $\lambda_\pm$.

		The  eigenvector matrix is 
		$$E_{\pm}=I+\frac{i}{k-\lambda_{\pm}}\sigma_3 Q_{\pm}$$ corresponding to $\mp i\lambda_{\pm}$
		and $\det E_\pm=-\frac{2\lambda_\pm}{k-\lambda_\pm}\triangleq \gamma_\pm(k)$.
		
		After diagonalizition,  we have 
		\begin{equation}
			\phi_{\pm}\sim E_{\pm}e^{-i\lambda_{\pm}x\sigma_3},\quad x\rightarrow \pm\infty.
		\end{equation}
		Define the normalized Jost functions $\mu_{\pm}$ as
		\begin{equation} 
			\mu_{\pm}=\phi_{\pm}e^{i\lambda_{\pm}x\sigma_3},
			\label{2}
		\end{equation}
		then we have $$\mu_{\pm}\rightarrow E_{\pm},\quad x\rightarrow \pm \infty.$$ 
		And they satisfy the Volterra integral equation
		\begin{equation}
			\mu_{\pm}=E_{\pm}(k)+\int_{\pm\infty}^x E_{\pm}(k)e^{-i\lambda_{\pm}(x-y)\hat\sigma_3} \left[
			E_{\pm}^{-1}(k)\Delta Q_{\pm}(y) \mu_{\pm}(y,k) 
			\right] \d y
			\label{integral}
		\end{equation}
		with $\Delta Q_{\pm}=Q-Q_{\pm}$.
		
		\begin{lemma}\label{2.1}  If $(1+|x|)(q-q_{\pm})\in L^1(\R_{\pm})$, then
			\begin{itemize}
				
				\item 
				
				Analyticity: the function $\mu_{-,1}(x,t,k)$ is well-defined in $\Sigma_-$ and analytic in $\C\backslash\Sigma_-$, whereas the function $\mu_{+,2}(x,t,k)$ is well-defined in $\Sigma_+$ and analytic in $\C\backslash\Sigma_+$.
				
				%		\item The first symmetry:
				%		\begin{equation}
					%			\phi_\pm(x,t,k)=\sigma_1 \overline{\phi_\pm(x,t,\bar k)}\sigma_1,
					%		\end{equation}
				\item Symmetry: Let $\widetilde \phi_{\pm}=\phi_{\pm}(x,t,-\lambda_{\pm}(k))$, we have
				\begin{equation}
					\phi_{\pm}=\widetilde \phi_{\pm}\frac{i}{k-\lambda_{\pm}}\sigma_3 Q_{\pm},\quad k\in\mathring{\Sigma}_{\pm},
				\end{equation}
				Column-wise reads
				\begin{align}
					&\phi_{\pm,1}=\frac{-iq_{\pm}}{k-\lambda_\pm}\widetilde \phi_{\pm,2}, \quad k\in\mathring{\Sigma}_{\pm},\label{2.6}\\
					&\phi_{\pm,2}=\frac{iq_{\pm}}{k-\lambda_\pm}\widetilde \phi_{\pm,1}, \quad k\in\mathring{\Sigma}_{\pm}.\label{2.7}
				\end{align}
				
				%		\item Asymptotics:
				%		\begin{align}
					%			&\mu_\pm(x,t,k)=I+O(\frac 1k), & k\rightarrow\infty, k\in \C^+,\\
					%			&\mu_\pm(x,t,k)=\frac{i}{k-\lambda_\pm}\sigma_3Q_\pm+O(1),  \hspace{0.2cm} &k\rightarrow\infty, k\in \C^-.
					%		\end{align}
			\end{itemize}
		\end{lemma}

		\subsubsection{Scattering matrix}
		There exists a scattering matrix $S(k,t)$ such that 
		\begin{equation}
			\phi_-(x,t,k)=\phi_+(x,t,k)S(k,t),\quad k\in\Sigma_-.
			\label{scattering}
		\end{equation}
		Thus, we have $\det S=\frac{\gamma_-}{\gamma_+}$. And the  scattering coefficients can be written as the Wronskians expressions
		\begin{equation}
			\begin{split}
				&s_{11}=\frac{\Wr(\phi_{-,1},\phi_{+, 2})}{\gamma_+},\quad  s_{12}=\frac{\Wr(\phi_{-,2},\phi_{+, 2})}{\gamma_+},\\
				&s_{21}=\frac{\Wr(\phi_{+,1},\phi_{-, 1})}{\gamma_+},\quad  s_{12}=\frac{\Wr(\phi_{+,1},\phi_{-, 2})}{\gamma_+}.
			\end{split}
		\end{equation}
		Define
		\begin{equation}
			\rho=\frac{s_{21}}{s_{11}},\quad \tilde \rho=\frac {s_{12}} {s_{22}}.
		\end{equation}
		\begin{lemma}\label{2.2}
		 Let $s_{11}^-(k)=\lim_{\epsilon\uparrow 0} s_{11}(k+i\epsilon, t)$, then we have 
				\begin{equation}
					\begin{split}
						&s_{11}^-(k)=s_{22}\frac{( k+\lambda_+ )(k-\lambda_-)}{  q_+q_- },\quad k\in\Sigma_-,\\
						&s_{11}^-(k)=s_{21}(k)\frac{k+\lambda_+}{-iq_+},\quad k\in\Sigma_0.
					\end{split}\label{2.17}
				\end{equation}
				%		\item $$S(k)=I+O(\frac 1k),\quad k\rightarrow\infty, k\in\C^+,$$
				%		$$S(k)=\frac{q_+}{q_-}I+O(k),\quad k\rightarrow\infty, k\in\C^-.$$
		\end{lemma}
		
		%	\begin{proof}
			%	Again, we only prove the second statement.  For $k\in \Sigma_-$, we know  $\gamma_+^-\triangleq \lim_{\epsilon\uparrow 0}\gamma_+(k+i\epsilon)=\frac{2\lambda_+}{k+i\lambda_+}$. It follows from the Wronskians expression of $s_{11}$, (\ref{2.6}) and (\ref{2.7}) that 
			%	\begin{equation}
				%		\begin{split}
					%			s_{11}^-=&\frac{\Wr (\phi_{-,1}^-, \phi_{+,2}^-)}    { \gamma_+^- }\\
					%			=&\frac{k+\lambda_+ }  {2k} \Wr \left(   \frac{k-\lambda_-}{iq_-} \phi_{-,2},  \frac{k-\lambda_+}   {iq_+}\phi_{+,1}       \right)\\
					%			=&s_{22}  \frac{( k+\lambda_+ )(k-\lambda_-)}{  q_+q_- },\quad k\in\Sigma_-.
					%		\end{split}
				%	\end{equation}
			%	For $k\in\Sigma_0$,  the function $\phi_{-,1}$ have no jump, i.e, $\phi_{-,1}^-=\phi_{-,1}$, thus we have 
			%	\begin{equation}
				%		\begin{split}
					%			s_{11}^-=&\frac{k+\lambda_+}  {2k}\Wr (\phi_{-,1}, \phi_{+,2}^-)\\
					%			&=\frac{k+\lambda_+}  {2k}\Wr (\phi_{-,1},  \frac{k-\lambda_+}{iq_+}\phi_{+,1}  )\\
					%			&=s_{21}\frac{k+\lambda_+}  {-iq_+}.
					%		\end{split}
				%	\end{equation}
			%\end{proof}

			\subsubsection{Discrete spectrum and residue condition}
			
			Assume that there exists no spectral singularities. We can prove that the discrete eigenvalues don't exist in the interval $\mathring\Sigma_0$, also can not be the branch points $\pm q_\pm$ (see more details in \cite{biondini}). And we assume 
			$s_{11}(k,t)$ has $J$ simple zeros in $(-q_+, q_+)$ denoted by $k_j,\hspace{0.1cm} j=1,2,\ldots, J$. From the Wronskians expression of $s_{11}(k,t)$, there exists a constant only depending on $t$ such that
			$$\phi_{-,1}(x,t,k_j)=b_j(t)\phi_{+,2}(x,t,k_j).$$
			On account of (\ref{2}), we have
			$$\mu_{-,1}(x,t,k_j)=b_j(t)\mu_{+,2(x,t,k_j)}e^{i(\lambda_{+,j}+\lambda_{-,j})x},$$
			where $\lambda_{\pm,j}=\lambda_\pm(k_j)$.
			Thus for $j=1,2,\ldots,J$, we get
			\begin{equation}
				\Res_{k=k_j}\bigg[   \frac{\mu_{-,1}(x,t,k)}{s_{11}(k,t)}    \bigg]=\frac{\mu_{-,1}(k_j)}{s_{11}^{'}(k_j)}=C_j(t)e^{i(\lambda_{+,j}+\lambda_{-,j})x}\mu_{+,2}(x,t,k_j),
			\end{equation}
			with $C_j=\frac{b_j}{s_{11}^{'}(k_j)}$ .
			Therefore, we have
			\begin{equation}
				\Res_{k=k_j} M(x,t,k)= \bigg(   
				C_j e^{i(\lambda_{+,j}+\lambda_{-,j})x} M_2(k_j) \hspace{0.1cm}, \hspace{0.1cm} 0
				\bigg).
			\end{equation}
		
		\subsection{Time evolution}
		Performing a similar maniputation of the time evolution \ref{timeevolution} in the focusing case, we obtain 
 $$f_\pm(k)=-2(q_\pm^2+2k^2)\lambda_\pm(k).$$
and 
\begin{equation}
	\begin{split}
		&s_{11}(k,t)=s_{11}(k,0)\exp[it(f_+(k)-f_-(k))],\\
		&s_{21}(k,t)=s_{21}(k,0)\exp[-it(f_+(k)+f_-(k))],
	\end{split}
\end{equation}
$$\rho(k,t)=\rho(k,0)\exp(-2itf_+(k)),$$
$$s_{11}^{'}(k_j,t)=s_{11}^{'}(k_j,0)\exp[it(f_+(k_j)-f_-(k_j))],$$
$$b_j(t)=b_j(0)\exp[-it(f_+(k_j)+f_-(k_j))].$$
Therefore we obtain
$$C_j(t)=C_j(0)\exp(-2if_+(k_j)t).$$

			\subsection{Inverse scattering problem}
			Next, we will establish the RH problem associated to the defocusing modified KdV equation (\ref{equ}) and recover the potential.
			\subsubsection{Riemann-Hilbert problem}
			Define a meromorphic function as 
			\begin{equation}
				M(x,t,k)=	\left(
				\frac{\mu_{-,1}}{s_{11}}\hspace{0.1cm}, \hspace{0.1cm} \mu_{+,2}
				\right),\quad k\in\Sigma_+.
				\end{equation}

			\noindent \textbf{Jump on the $\Sigma_-$} 
			Using  the scattering relation (\ref{scattering}),  the second symmetry (\ref{2.6}) (\ref{2.7}) of the Jost function and (\ref{2.17}) of the scattering coefficients, we have 
			\begin{equation}
				\begin{split}
					\frac{\phi_{-,1}}{s_{11}}=&(1-\rho\tilde \rho)\phi_{+,1}+ \rho\frac{\phi_{-,2}}{s_{22}}\\
					=&(1-\rho \tilde \rho) \frac{-iq_+}{k-\lambda_+} \widetilde \phi_{+,2}+ \rho \frac{(k+\lambda_+)(k-\lambda_-)}{q_+q_-}\frac{iq_-}{k-\lambda_-} \frac{\tilde \phi_{-,1}}{s_{11}^-}\\
					=&(1-\rho \tilde \rho) \frac{-iq_+}{k-\lambda_+} \widetilde \phi_{+,2}+\rho \frac{i(k+\lambda_+)}{q_+}\frac{\widetilde \phi_{-,1}}{s_{11}^-}
				\end{split}\label{2.24}
			\end{equation}
			and 
			\begin{equation}
				\begin{split}
					\phi_{+,2}=&\frac{\phi_{-,2}}{s_{22}}-\tilde \rho \phi_{+,1}\\
					=&\frac{(k+\lambda_+)(k-\lambda_-)}{q_+q_-} \frac{iq_-}{k-\lambda_-}   \frac{\widetilde \phi_{-,1}}{s_{11}^-}
					-\tilde \rho\frac{ -iq_+}{k-\lambda_+}\widetilde \phi_{+,2}\\
					=&\frac{i(k+\lambda_-)}{q_+} \frac{ \tilde \phi_{-,1} }   { s_{11}^-  }+\tilde \rho \frac{iq_+}{k-\lambda_+}\widetilde \phi_{+,2}.
				\end{split}\label{2.25}
			\end{equation}
			Combining (\ref{2.24}) and (\ref{2.25}) leads to 
			\begin{equation}
				\left(   \frac { \phi_{-,1}  }   { s_{11} }\hspace{0.1cm},\hspace{0.1cm} \phi_{+,2}   \right)= 	\left(   \frac { \phi_{-,1}^-   }   { s_{11}^- }\hspace{0.1cm},\hspace{0.1cm} \phi_{+,2}^-    \right)
				\left(\begin{array}{cc}
					\rho \frac{i(k+\lambda_+)}{q_+}  &  \frac{i(k+\lambda_+)}{q_+}\\
					(1-\rho\tilde\rho)\frac{-iq_+}{k-\lambda_+} & \tilde\rho\frac{iq_+}{k-\lambda_+}
				\end{array}\right)
				,\quad k\in\Sigma_-.
			\end{equation}
			\noindent \textbf{Jump on the $\Sigma_0$} According to (\ref{2.17}), we get
			\begin{equation}
				\frac{\phi_{-,1}}{s_{11}}=\frac{\phi_{-,1}}{s_{11}^-}\frac{s_{11}^-}{s_{11}}=\rho \frac{k+\lambda_+}{-iq_+}\frac{\phi_{-,1}}{s_{11}^-}.
				\label{2.27}
			\end{equation}
			Again, using the scattering relation (\ref{scattering}) and (\ref{2.17}), we have
			\begin{equation}
				\begin{split}
					\phi_{+,2}=&\frac{1}{s_{21}}\phi_{-,1}-\frac{1}{\rho}\phi_{+,1}\\
					=&\frac{k+\lambda_+}{-iq_+}\frac{\phi_{-,1}}{s_{11}^-}-\frac{1}{\rho} \frac{-iq_+}{k-\lambda_+}\widetilde \phi_{+,2}.
				\end{split}\label{28}
			\end{equation}
			Combining (\ref{2.27}) and (\ref{28}) leads to 
			\begin{equation}
				\left(   \frac { \phi_{-,1}  }   { s_{11} }\hspace{0.1cm},\hspace{0.1cm} \phi_{+,2}   \right)= 	\left(   \frac { \phi_{-,1}   }   { s_{11}^- }\hspace{0.1cm},\hspace{0.1cm} \phi_{+,2}^-    \right)
				\left(\begin{array}{cc}
					\rho \frac{(k+\lambda_+)}{-iq_+}  &  \frac{(k+\lambda_+)}{-iq_+}\\
					0 & \frac{1}{\rho}\frac{iq_+}{k-\lambda_+}
				\end{array}\right)
				,\quad k\in\Sigma_0.
			\end{equation}
			Note the relation for $\mu_{\pm}$ and $\phi_\pm$ and  calculate direct, we arrive at
			\begin{equation}
				\left(
				\frac{\mu_{-,1}}{s_{11}}\hspace{0.1cm}, \hspace{0.1cm} \mu_{+,2}
				\right)=\left\{ \begin{aligned}
					& 	\left(	\frac{\mu_{-,1}^-}{s_{11}^-}\hspace{0.1cm}, \hspace{0.1cm} \mu_{+,2}^-
					\right) V_{\Sigma_-},\quad k\in\Sigma_-
					\\
					& \left(	\frac{\mu_{-,1}}{s_{11}^-}\hspace{0.1cm}, \hspace{0.1cm} \mu_{+,2}^-
					\right) V_{\Sigma_0},\quad k\in\Sigma_0
				\end{aligned}\right.
			\end{equation}
			where 
			\begin{equation}
				\begin{split}
					&V_{\Sigma_-}=(E_+-I)\left(\begin{array}{cc}
						e^{-i\lambda_+x} &  \\
						&   e^{i\lambda_-x}
					\end{array}\right)
					\left(\begin{array}{cc}
						1-\rho\tilde\rho  &  -\tilde\rho\\
						\rho & 1 
					\end{array}\right)
					\left(\begin{array}{cc}
						e^{i\lambda_-x} &  \\
						&   e^{-i\lambda_+x}
					\end{array}\right),\\%???
					&V_{\Sigma_0}=(E_+-I)\left(\begin{array}{cc}
						e^{-i\lambda_+x} &  \\
						&   e^{-i\lambda_-x}
					\end{array}\right)
					\left(\begin{array}{cc}
						0  &  \frac{1}{\rho}\\
						\rho & - 1 
					\end{array}\right)
					\left(\begin{array}{cc}
						e^{i\lambda_-x} &  \\
						&   e^{-i\lambda_+x}
					\end{array}\right).
				\end{split}\nonumber
			\end{equation}
			Therefore, we obtain a jump condition 
			\begin{equation}
				M^+=M^-(E_+-I)(I-V_0),\quad k\in\Sigma_+,
				\label{31}
			\end{equation}
			with
			\begin{equation}
				V_0=\left\{\begin{aligned}%???
					&\left(\begin{array}{cc}
						1-(1-\rho\tilde\rho)e^{i(\lambda_--\lambda_+)x} & \tilde\rho e^{-2i\lambda_+} \\
						\rho e^{2i\lambda_-} & 1-e^{-i(\lambda_--\lambda_+)x}
					\end{array}\right),\quad k\in\Sigma_-\\%???
					&\left(\begin{array}{cc}
						1  & \frac{1}{\rho}e^{-2i\lambda_+} \\
						-\rho & 1-e^{-i(\lambda_++\lambda_-)x}
					\end{array}\right),\quad k\in\Sigma_0
				\end{aligned}	\right.
				\nonumber
			\end{equation}
			Due to the different asymptotics as $k\rightarrow\infty$ between $k\in\C^+$  and  $k\in\C^-$
			\begin{equation}
				\lambda_{\pm}=\left\{
				\begin{aligned}
					k-\frac{q_\pm^2}{2k}+o\bigg(\frac 1 k\bigg),\quad &k\rightarrow \infty,\quad \Im k>0,\\
					-k+\frac{q_\pm^2}{2k}+o\bigg(\frac 1 k\bigg),\quad &k\rightarrow \infty,\quad \Im k<0,
				\end{aligned} 
				\right.
			\end{equation}
			we obtain the asymptotic of $M(x,t,k)$
			\begin{equation}
				M(x,t,k)=\left\{
				\begin{aligned}
					&I+O(\frac 1 k),\quad & k\rightarrow \infty,\quad \Im k>0,\\
					&\frac{i}{k-\lambda_+}\sigma_3Q_++O(1),\quad & k\rightarrow \infty,\quad \Im k<0.
				\end{aligned}
				\right.
			\end{equation}
			which can be rewritten as
			$$ M(x,t,k)=E_++O(1),\quad k\rightarrow \infty.$$
			To deal with the different limit, we introduce a new function $N(x,t,k)$ 
			\begin{equation}
				M(x,t,k)=N(x,t,k)E_+(k,t),
				\label{34}
			\end{equation}
			thus we have 
			$$N(x,t,k)=I+O(\frac 1k).$$
			On account of the jump (\ref{31}), we have
			\begin{equation}
				N^+=N^-\tilde V,\quad k\in\Sigma_+,
			\end{equation}
			where $\tilde V=E_+^-VE_+^{-1}$.
			\subsubsection{Reconstruction formula}
			Using the Plemelj formula, we can express the solution $M$ as follows:
			\begin{equation}
				\begin{split}
					N(x,t,k)=&I+\sum_{j=1}^J\frac{1}{k-k_j}\Res_{k=k_j}N(x,t,k)\\
					&-\frac{1}{2\pi i}\int_{\Sigma_+}\frac{[N^-(I-\tilde V)](x,t,\zeta)}{\zeta-k}\d \zeta,\quad k\in\C\backslash\Sigma_+,
				\end{split}
			\end{equation}
			from (\ref{34}), we obtain
			\begin{equation}
				\begin{split}
					M(x,t,k)=&E_+(k,t)+\sum_{j=1}^J\frac{1}{k-k_j}\Res_{k=k_j}M(x,t,k) E_+^{-1}(k_j) E_+(k,t)\\
					&-\frac{E_+(k,t)}{2\pi i}\int_{\Sigma_+}\frac{[M^-(E_+-I)V_0E_+^{-1}](x,t,\zeta)}{\zeta-k}\d \zeta,\quad k\in\C\backslash\Sigma_+,
				\end{split}
			\end{equation}
		By 	substituting the residue conditions, we have 
			\begin{equation}
				\begin{split}
					M(x,t,k)=&E_+(k,t)+\sum_{j=1}^J C_j(t)e^{i  (\lambda_{+,j}+\lambda_{-,j})  x   }
					\left(  M_2(k_j)  \hspace{0.1cm}, \hspace{0.1cm} 0   \right) E_+^{-1}(k_j,t)E_+(k,t)\\
					&+\frac{1}{2\pi i k}\int_{\Sigma_+}  [M^-(E_+-I)V_0E_+^{-1}](x,t,\zeta)\d \zeta+O\bigg(\frac{1}{k^2}\bigg),\\
					&k\rightarrow\infty, \quad \Im k>0.
				\end{split}\label{38}
			\end{equation}
			By a standard argument, we can recover the potential $q(x,t)$ 
			\begin{equation}
				q(x,t)=-i\lim_{\substack{
						k\rightarrow\infty\\
						\Im k>0}}kM_{12}(x,t,k),	
			\end{equation} 
			by substituting (\ref{38}) into it, we obtain
			\begin{equation}
				\begin{split}
					q(x,t)=&q_+(t)\bigg(1-\sum_{j=1}^J C_j(t)e^{i  (\lambda_{+,j}+\lambda_{-,j})  x   }M_{12}(x,t,k_j)
					\bigg)\\
					&+\frac{1}{2\pi }\int_{\Sigma_+}\frac{1}{\lambda_+(k)}\bigg[     
					(       \frac{iq_+}{k+\lambda_+}V_{0,11}-V_{0,12}   ) q_+(t)M_{12}^-(x,t,k)   )\\
					& -(   \frac{iq_+}{k+\lambda_+}V_{0,21}+V_{0,22}   ) q_+(t)M_{11}^-(x,t,k)	
					\bigg]\d k
				\end{split}
			\end{equation}
			\noindent\textbf{Remark 3.1} 	 When $q_-=q_+=q_0$, the reconstruction formula (\ref{rec}) matches Eq.  (128)  in \cite{zhangyan}.

			\noindent\textbf{Acknowledgements}
			
			This work is supported by  the National Science
		Foundation of China (Grant No.12371247).\vspace{2mm}
			
			\noindent\textbf{Data Availability Statements}
			
			The data that supports the findings of this study are available within the article.\vspace{2mm}
			
			\noindent{\bf Conflict of Interest}
			
			The authors have no conflicts to disclose.

		\end{document}